\documentclass[11pt]{article}
\usepackage{amsfonts,amsmath,amsthm,amssymb,amsopn,graphicx}
\usepackage{dsfont,bbm}
\usepackage{enumitem,multibib}
\usepackage{mathrsfs} 
\usepackage{hyperref}

\newcommand{\ii}{ {\rm i} }
\def\bra#1{\mathinner{\langle{#1}|}}
\def\ket#1{\mathinner{|{#1}\rangle}}
\def\abs#1{\left | #1 \right |}
\usepackage{wordlike}
\usepackage{newfloat}
\usepackage{pdfpages}
\pdfoutput=1
\DeclareFloatingEnvironment[name={Supplementary Figure}]{suppfigure}
\PassOptionsToPackage{margin=1in}{geometry}

\usepackage[hang,flushmargin]{footmisc} 

\usepackage[parfill]{parskip}
\usepackage{setspace}
\usepackage{fancyhdr}

\let\myfootnote\footnote
\renewcommand{\footnote}[1]{\myfootnote{~#1}}

\raggedright

\newcommand{\ketbra}[2]{\vert  #1\rangle   \langle #2 \vert}

\newcommand{\ave}[1]{{\langle #1\rangle}}
\newcommand{\z}{{\rm z}}

\newcommand{\tr}{\rm{tr}}
\newcommand{\RaR}{\mathbb{R}}
\def\one{\mathbbm{1}}
\def\bra#1{\mathinner{\langle{#1}|}}
\def\ket#1{\mathinner{|{#1}\rangle}}

\def \dH{{\mathfrak{H}}}

\newtheorem{corollary}{Supplementary Corollary}

\usepackage{authblk}

\newcommand{\dd}{ {\rm d} }

\newcommand{\x}{{\rm x}}

\newcommand{\LL}{{{\cal L}}}

\def\tr{{{\rm tr}}}

\def\End{{\,{\rm End}\,}}
\def\one{\mathbbm{1}}
\def\Re{{\,{\rm Re}\,}}
\def\im{{\,{\rm Im}\,}}

\def\Re{{\,{\rm Re}\,}}
\def\im{{\,{\rm Im}\,}}

\newcommand*{\citen}[1]{%
	\begingroup
	\romannumeral-`\x 
	\setcitestyle{numbers}%
	\cite{#1}%
	\endgroup   
}
\usepackage{color}

\definecolor{wordblue}{RGB}{9, 112, 192}

\newtheorem{theorem}{Supplementary Theorem}

\title{{\textnormal{\color{wordblue}Supplementary Information: Non-stationary coherent quantum many-body dynamics through dissipation}}}
\author[1]{Bu\v{c}a, et al.} 
\date{}
\usepackage[superscript,biblabel]{cite}
\begin{document}
	

# Non-stationary coherent quantum many-body dynamics through dissipation

Berislav Buča[1], Joseph Tindall[1], and Dieter Jaksch[1,2]

Emails: berislav.buca@physics.ox.ac.uk, joseph.tindall@physics.ox.ac.uk, dieter.jaksch@physics.ox.ac.uk

[1]Clarendon Laboratory, University of Oxford, Parks Road, Oxford OX1 3PU, United Kingdom
[2]Centre for Quantum Technologies, National University of Singapore, 117543 Singapore

**The assumption that quantum systems relax to a stationary state in the long-time limit underpins statistical physics and much of our intuitive understanding of scientific phenomena. For isolated systems this follows from the eigenstate thermalization hypothesis. When an environment is present the expectation is that all of phase space is explored, eventually leading to stationarity. Notable exceptions are decoherence-free subspaces that have important implications for quantum technologies and have so far only been studied for systems with a few degrees of freedom. Here we identify simple and generic conditions for dissipation to prevent a quantum many-body system from ever reaching a stationary state. We go beyond dissipative quantum state engineering approaches towards controllable long-time non-stationarity typically associated with macroscopic complex systems. This coherent and oscillatory evolution constitutes a dissipative version of a quantum time-crystal.. We discuss the possibility of engineering such complex dynamics with fermionic ultracold atoms in optical lattices.**

Introduction

The eigenstate thermalization hypothesis[1,2] (ETH) states that an isolated many-body quantum system with non-integrable Hamiltonian relaxes locally to a stationary equilibrium ensemble. For generic initial states local observables are given by thermal expectation values after a sufficiently long evolution time $t$. A generalized ETH holds if the system is integrable or under the influence of weak integrability breaking[1,2] Equilibration occurs on relatively short timescales, typically within a few characteristic periods.

Perfect isolation is impossible in experiments and interactions with the environment will always provide additional relaxation mechanisms. The widely used – but notoriously difficult to prove – assumption of ergodicity states that even weak coupling to an environment enables the system to explore the entire connected non-decaying part of the system Hilbert space $\mathcal{H}$ as sketched in Fig. 1a). The evolution thus induces relaxation to a unique stationary state $\rho_\infty$ in the long-time limit.

In quantum technology platforms[3] a microscopic understanding of the environmental coupling allows control of the open system dynamics $\dot{\rho}(t) = \mathcal{L}\rho(t)$ of the density operator $\rho(t)$. The super-operator $\mathcal{L}$ can be engineered to possess a small number of controllable purely imaginary eigenvalues. The corresponding eigenstates are protected from the environment and form a decoherence free sub-space[4,5] where quantum information can be processed without leaking into the environment. Controlled dissipation can also lead to many-body pure states that are stationary[6,7,8].

The seemingly robust feature of relaxation to stationarity in quantum many-body systems presents a puzzle when contrasted with the emergence of non-stationary dynamics often observed in macroscopic systems[9,10]. Non-stationarity plays an important role in many areas ranging from

microbiology[11,12] and neurobiological systems[9,13] across climate science[14,15] to financial time series[9,16]. It remains almost unstudied in quantum statistical physics where research of non-equilibrium setups mostly concentrates on currents of time-independent quantities. The question thus arises whether generic insights into the microscopic origins of non-stationary and complex long-time evolution may be gleaned from the study of highly controlled and well understood experiments in the quantum regime.

Here we show that coupling to an environment can induce non-stationarity in many-body quantum systems that would otherwise relax, through mutual dephasing of its eigenstates, according to the ETH. Symmetry-preserving dissipation eliminates a large class of eigenstates and ensures constructive interference. It splits the non-decaying part of the Hilbert space into disjoint sectors schematically shown in Fig. 1b). In the long-time limit a dark Hamiltonian coherently drives the system between these disjoint parts leading to non-decaying oscillations in observables that are not entirely contained in one sector. We will give general conditions guaranteeing this situation and study an example realizable in current experiments with ultracold atoms.

**Results**

*Conditions for non-stationarity in a many-body system*

Specifically, our starting point is the Lindblad master equation modelling a quantum system weakly coupled to an environment that acts as a source of noise. The main results presented here are also valid for open quantum systems beyond the Lindblad framework (see Supplementary Methods for details). The master equation is given by (setting $\hbar = 1$)

$$\dot{\rho}(t) = \mathcal{L}\rho = -i[H, \rho(t)] + \sum_\mu \left(2 L_\mu \rho L_\mu^\dagger - L_\mu^\dagger L_\mu \rho - \rho L_\mu^\dagger L_\mu\right), \quad (1)$$

where the first term describes unitary evolution $i\dot{\rho} = [H, \rho(t)] = H\rho(t) - \rho(t)H$ of an isolated system with Hamiltonian $H$ and follows directly from the Schrödinger equation. The second term contains the jump operators $L_\mu$ arising from decoherence processes induced by the environment. Formally, the density operator will be non-stationary if the Liouvillian $\mathcal{L}$ has purely imaginary eigenvalues[17,18,19] $\mathcal{L}\rho_n = -i\mathfrak{H}\rho_n = -i\lambda_n \rho_n$ for eigenoperators $\rho_n$. Here we have defined the dark Hamiltonian $\mathfrak{H}$ as the part of the evolution that is purely coherent.

The conceptually simplest situation where non-stationarity may occur is well understood for systems with few degrees of freedom[18,19,20]. All jump operators fulfil $L_\mu|\phi_n\rangle = 0$ for a subset of eigenstates $|\phi_n\rangle$ with eigenvalues $\omega_n$ of the Hamiltonian. These so-called dark states are perfectly decoupled from the environment and span a decoherence free subspace[4]. Coherences between dark states evolve according to $\mathcal{L}|\phi_n\rangle\langle\phi_m| = i(\omega_m - \omega_n)|\phi_n\rangle\langle\phi_m|$ and undergo continued oscillations induced by the coherent part of the dynamics. The dark Hamiltonian $\mathfrak{H}$ may then be understood as a purge of unwanted eigenstates of the original Hamiltonian $H$. We give an example of a many-body dark Hamiltonian in the Supplementary Discussion.

A dark Hamiltonian is not required to be Hermitian and its eigenstates need not be pure. We concentrate on this more general and interesting case and show that it may lead to non-stationary and complex long-time dynamics. This case is realized if there exists an eigenoperator $A$ such that (see Supplementary Methods for details)

$$[H, A] = -\lambda A \quad \text{and} \quad [L_\mu, A] = [L_\mu^\dagger, A] = 0 \quad \forall\, \mu, \quad (2)$$

with real valued $\lambda$. We find $\mathcal{L}\rho_{nm} = \mathrm{i}(m-n)\lambda\rho_{nm}$ for operators $\rho_{nm} = A^n \rho_\infty (A^\dagger)^m$ and integer $m, n > 0$. Here, $\rho_\infty$ is a stationary state with $\mathcal{L}\rho_\infty = 0$. The fact that $A$ is an eigenoperator and not just a symmetry with $[H, A] = 0$ is crucial and guarantees that the operators $\rho_{nm}$ are not stationary for $n \neq m$. We refer to these $\rho_{nm}$ as mixed coherences because they describe oscillations induced by $\mathfrak{H}$ between, usually mixed, stationary states $\rho_{nn}$ [see Fig. 1b)]. In contrast to the coherences in a decoherence free subspace they are not decoupled from the environment and are affected by dissipation $L_\mu \rho_{nm} L_\mu^\dagger \neq 0$. All initial states that contain mixed coherences $\rho_{nm}$ will continuously oscillate in the long-time limit. If only one operator $A$ exists then the spectrum of the dark Hamiltonian is equidistant, like that of a harmonic oscillator. The equidistance of the spectrum ensures the long-time dynamics is periodic, with period $2\pi/\lambda$ and the system does not relax to stationarity.

*Non-stationary dynamics in the open Hubbard Model*

We study the emergence of non-stationarity in $D$-dimensional fermionic Hubbard models that possess spin and $\eta$-pairing symmetries (see Methods). This is a paradigmatic example that can be accurately realized in highly controllable quantum systems, such as optical lattices filled by ultracold spin 1/2 atoms[21]. The Hamiltonian is

$$H = -\tau \sum_{\langle j,j'\rangle, s} c_{j,s}^\dagger c_{j',s} + c_{j',s}^\dagger c_{j,s} + \sum_j U\, n_{j,\uparrow} n_{j,\downarrow} + \epsilon_j n_j + \frac{B}{2}(n_{j,\uparrow} - n_{j,\downarrow}), \quad (3)$$

where $\langle j, j'\rangle$ denotes nearest-neighbour sites of a bipartite lattice with $M$ sites and $c_{j,s}$ is the annihilation operator for a fermion with spin $s$ on site $j$. The particle number operator is $n_{j,s} = c_{j,s}^\dagger c_{j,s}$ and $n_j = n_{j,\uparrow} + n_{j,\downarrow}$. The hopping amplitude is $\tau$, $U$ denotes onsite interactions and $\epsilon_j$ is a site dependent energy offset. In an optical lattice, the term $\epsilon_j$ describes the trapping potential and/or spin-agnostic disorder e.g. created through speckle patterns[3,22]. A constant external magnetic field splits different spin states by $B$ via the Zeeman effect. We assume the coupling of the Hubbard lattice to the environment to take the form of local dephasing Lindblad operators $L_j = \gamma_j n_j$. In optical lattices, this can be achieved e.g. through immersion into a Bose-Einstein condensate[23] (see Supplementary Discussion for details).

The strong symmetries[17,19,24] of this model determine its generalized grand-canonical-like equilibrium states* as $\rho_\infty \propto \exp(\beta_0 N + \beta_1(S^+ S^-) + \beta_2 S^z)$, where $N$ is the total number of particles and $\vec{S} = (S^x, S^y, S^z)$ the total spin. The parameters $\beta_i$ play the role of generalized chemical potentials determined by the initial state. The operator $S^+$ fulfils the criteria of the eigenoperator $A$ with $\lambda = B$ and hence constructs a dark Hamiltonian $\mathfrak{H}$ (see Methods).

We study the system evolution starting from non-correlated polarized initial states. In Fig. 2a) we show the bulk-averaged fermion spin along the $x$-direction $\langle\langle S_i^x(t)\rangle\rangle$. The long-time oscillation amplitude of spin-spin correlations $\langle S_i^x(t) S_{i+j}^x(t)\rangle$ for arbitrary $i$ and $j$ are shown in Fig. 2b). After a short transient time these observables start oscillating with an amplitude quickly converging to a finite value with increasing $M$. Their spectrum is then narrowly centred around multiples of $B$ as shown in Fig. 2c). This is in excellent agreement with the analytically expected purely sinusoidal evolution in the long-time limit. In Fig. 2d) we compare traces of the spin dynamics in the $xy$-plane for different initial spin polarizations. All realizations (see Methods) of the stochastic dynamics are identical for the maximally polarized state, which thus behaves similarly to an isolated collection of

---

* Note that the stationary subspace is degenerate.

non-interacting spins. However, realizations for non-maximally polarized states possess fluctuations that increase with system size $M$. Only after averaging many realizations perfectly sinusoidal oscillations emerge following the initial transient. This evolution strongly violates ergodicity and is qualitatively different from the precession of independent spins.

In Fig. 3a) we study a quench starting from the ground state of the Hubbard model. In the absence of dephasing, the combination of disorder and many-body thermalization quickly dampens out the dynamics, as shown in the inset of Fig 3a). The closed system exhibits small fluctuations following revivals due to finite-size effects. Remarkably, in the presence of dephasing[22,23] persistent spin oscillations with frequency $B$ ensue after the quench. The strength of the system environment coupling solely determines the time for the transient dynamics to decay and coherent, oscillatory behaviour appear in the measured observables. In Fig. 3b) we show that fundamentally quantum off-diagonal long-range order[25] in the spin sector $\lim_{n\to\infty} \langle S_i^+ S_j^- \rangle \neq 0, \forall\, i,j$ is constructed by the dephasing dynamics even when starting from high-temperature thermal states of the Hubbard model.

We emphasize that the eigenstates of the dark Hamiltonian $\mathfrak{H}$ driving these oscillations are mixed and cannot be realized in an isolated system. Furthermore, the system admits no decoherence free subspaces as any state $|\phi\rangle$ for which $L_j|\phi\rangle = 0, \forall j$ cannot be an eigenstate of $H$ for finite hopping $\tau$. Indeed, all coherences that lead to dephasing in the isolated system get damped out by the dissipation because the setup does not admit dark states (see Supplementary Methods for a more detailed discussion).

We apply well-established complexity measures based on entropy[10] (see Methods) to the time evolution induced by $\mathfrak{H}$. Figure 4a) shows the mutual information between lattice sites as a function of time. In the presence of dephasing we see that for small times this is uniform and large which indicates that the reduced quantum state of a single site contains a large amount of information about the rest of the system. During the time evolution the mutual information decreases while simultaneously the disparity, shown in Fig. 4b), increases. Even a relatively small system reaches a complex state with little mutual information and large disparity between different sites. This is consistent with Fig. 2d) showing large fluctuations in individual realizations of the evolution. The experimental characterization of such a state necessarily requires measuring many sites. Figure 4 also shows that this complexity does not emerge in the closed system.

**Discussion**

Starting from the Hubbard model, different couplings to the environment can realize different classes of dark Hamiltonians (see the Supplementary Discussion for the details). For instance, when $\epsilon_j = \epsilon$ spin dephasing $L_j = \gamma_j S_j^z$ results in a dark Hamiltonian whose eigenstates all possess long-range off-diagonal $\eta$-pairing order, i.e. $\lim_{n\to\infty} \langle \eta_i^+ \eta_j^- \rangle \neq 0, \forall\, i,j$. Spin dephasing could thus contribute to the formation of superconducting states by inducing $\eta$-pairing[26].

More generally, our results open up the possibility of studying quantum statistical physics[27] of non-Hermitian[28] dark Hamiltonians. Linear response theory, behaviour under periodic driving, relaxation towards subspaces of the dark Hamiltonian, the formulation of a semiclassical limit and metastability[29] are also interesting and open questions. The asymptotic coherent dynamics induced by a dark Hamiltonian breaks time-translation symmetry. It may thus be understood as the dissipative realization of a fully quantum time crystal[30,31] in the bulk that does not require external time-dependent driving[31] or collective dissipation of a non-interacting system[32].

We have shown that relaxation to equilibrium and stationarity can be prevented by environmental dissipation. This causes some degrees of freedom to dampen out and stops them from dephasing. The underlying physics resembles classical complex system dynamics where also not all available degrees of freedom contribute to the formation of collective complex behaviour[9].

**Methods**

*Symmetries of the Hubbard model*

The $D$-dimensional Hubbard Hamiltonian on a bipartite lattice commutes with two sets of generators of the $su(2)$ algebra. The first set consists of spin operators[33],

$$S^z = \sum_j S_j^z, \qquad S_j^z = \frac{1}{2}(n_{j,\uparrow} - n_{j,\downarrow}),$$

$$S^+ = \sum_j S_j^+, \qquad S_j^+ = c_{j,\uparrow}^\dagger c_{j,\downarrow},$$

$$S^- = \sum_j S_j^-, \qquad S_j^- = c_{j,\downarrow}^\dagger c_{j,\uparrow},$$

where $c_{j,\downarrow}$ ($c_{j,\uparrow}$) is the standard fermionic annihilation operator annihilating a down (up) spin on site $j$. We have,

$$[H, S^z] = 0, \qquad [H, S^\pm] = \pm B\, S^\pm.$$

The other, hidden, $SU(2)$ symmetry, called $\eta$-pairing, is given in terms of its generators as,

$$\eta^z = \frac{1}{2}\sum_j (n_j - 1),$$

$$\eta^+ = \sum_j \tau(j)\eta_j^+, \quad \eta_j^+ = c_{j,\uparrow}^\dagger c_{j,\downarrow}^\dagger,$$

$$\eta^- = \sum_j \tau(j)\eta_j^-, \quad \eta_j^- = c_{j,\downarrow} c_{j,\uparrow},$$

where $\tau(j)$ follows an alternating checkerboard pattern of $\pm 1$. With $\epsilon_j = \epsilon$ we have,

$$[H, \eta^z] = 0, \qquad [H, \eta^\pm] = \pm 2\epsilon \eta^\pm.$$

Crucially, we also have $\left[S_j^\alpha, \eta_k^\beta\right] = 0, \forall\, \alpha, \beta, j, k$. This fact allows us to construct Lindblad operators in terms of either spin or $\eta$-pairing operators and get dark Hamiltonians in the long-time limit. In the main text we study the example $L_j = \gamma_j\, n_j$. Explicitly, the local transverse magnetizations are given by $S_j^x = (S_j^+ + S_j^-)/2$ and $S_j^y = i(S_j^+ - S_j^-)/2$.

*Quantum mutual information and disparity*

Taking a complex network measure applied to quantum systems from Ref. 10 we study the complexity of the coherent dynamics using quantum mutual information,

$$I_{ij} = \frac{1}{2}(S_i + S_j - S_{ij}),$$

where $S_i = \text{tr}(\rho_i \log \rho_i)$ and $S_{ij} = \text{tr}(\rho_{ij} \log \rho_{ij})$ are the one- and two-point reduced von Neumann entropies of subsystems $\rho_i = \text{tr}_{k \neq i} \rho$ and $\rho_{i,j} = \text{tr}_{k \neq i,j} \rho$. Using this we also define the disparity $Y_i$,

$$Y_i = \frac{\sum_{j=1}^{M}(I_{ij})^2}{\left(\sum_{j=1}^{M} I_{ij}\right)^2},$$

which may intuitively be understood by observing that it is small when the quantum mutual information between site $i$ and the other sites takes on a constant value and large when one particular $I_{ij}$ takes on a dominant value. More specifically, we study the average disparity across the sites $Y = \frac{1}{M}\sum_{j=1}^{M} Y_j$.

*Simulation of the master equation*

The numerical calculations shown in Figs. 2 and 3a) were performed by a stochastic unravelling of the master equation into individual realizations by the quantum trajectories method[34]. The trajectories were calculated using the Tensor Network Theory Library[35]. In Figs. 3b) and 4) we numerically integrated the full matrix representation of the master equation directly.

*Data availability*

The data that supports the plots within this paper and other findings of this study are available from the authors upon reasonable request. The figures were produced with Python and processed with Inkscape.

*Code availability*

The Tensor Network Theory Library, which can be used to perform the simulations in the article, is available at http://www.tensornetworktheory.org/. The programming scripts used to obtain the data in this manuscript are available from the authors upon reasonable request.

**References**

[1] D'Alessio, L., Kafri, Y., Polkovnikov, A., & Rigol, M., From quantum chaos and eigenstate thermalization to statistical mechanics and thermodynamics. *Advances in Physics* **65**, 239 (2016).

[2] Rigol, M., Dunjko, V., & Olshanii, M. Thermalization and its mechanism for generic isolated quantum systems. *Nature* **452**, 854 (2008).

[3] Bloch, I., Dalibard, J., & Zwerger, W. Many-body physics with ultracold gases. *Rev. Mod. Phys*. **80**, 885 (2008).

[4] Lidar, D. A., Chuang, I. L., & Whaley, K. B., *Phys. Rev. Lett.* **81**, 2594 (1998).

[5] Knill, E., Laflamme, R., & Viola, L. Theory of Quantum Error Correction for General Noise. *Phys. Rev. Lett*. **84**, 2525 (2000).

[6] Diehl, S., Micheli, A., Kantian, A., Kraus, B., Büchler, H. P., & Zoller, P. Quantum states and phases in driven open quantum systems with cold atoms. *Nat. Phys.* **4**, 878 (2008).

[7] Verstraete, F., Wolf, M. M., & Cirac, J. I. Quantum computation and quantum-state engineering driven by dissipation. *Nat. Phys.* **5**, 633 (2009).

[8] N. Syassen, et al. Strong Dissipation Inhibits Losses and Induces Correlations in Cold Molecular Gases. *Science* **320**, 1329 (2008).

[9] Kwapieńa, J., & Drożdżab, S. Physical approach to complex systems. *Physics Reports* **515,** 115 (2012).

[10] Valdez, M.A., Jaschke, D., Vargas, D. L. & Carr, L.D. Quantifying Complexity in Quantum Phase Transitions via Mutual Information Complex Networks. *Phys. Rev. Lett*. **119**, 225301 (2017).

[11] Benincà, E., et al. Chaos in a long-term experiment with a plankton community. *Nature* **451**, 822 (2008).

[12] Hekstra, D.R., & Leibler, S. Contingency and Statistical Laws in Replicate Microbial Closed Ecosystems. *Cell* **149**, 1164 (2012).

[13] Kelso, J.A.S. *Dynamic Patterns: The Self-Organization of Brain and Behavior.* (MIT Press, Boston, 1995).

[14] Myhre, G., et al. *Climate Change 2013: The Physical Science Basis. Contribution of Working Group I to the Fifth Assessment Report of the Intergovernmental Panel on Climate Change.* (Cambridge University Press, Cambridge, 2013).

[15] F. Ji, et al. Evolution of land surface air temperature trend. *Nature Climate Change* **4**, 462 (2014).

[16] Tsay, R.S. *Analysis of Financial Time Series* (John Wiley, New York, 2002).

[17] Albert, V. V., Bradlyn, B., Fraas, M., & Jiang L. Geometry and Response of Lindbladians. *Phys. Rev. X* **6**, 041031 (2016).

[18] Baumgartner, B., & Narnhofer, H. Analysis of quantum semigroups with GKS–Lindblad generators: II. General. *J. Phys. A: Math. Theor.* **41**, 395303 (2008).

[19] Albert, V. V., & Jiang, L. Symmetries and conserved quantities in Lindblad master equations. *Phys. Rev. A* **89**, 022118 (2014).

[20] Choi, M.-D., & Kribs, D. W., Method to Find Quantum Noiseless Subsystems. *Phys. Rev. Lett.* **96**, 050501 (2006).

[21] Jaksch, D., & Zoller P. The cold atom Hubbard toolbox. *Annals of physics* **315**, 52 (2005).

[22] Lüschen, H.P., et al. Signatures of Many-Body Localization in a Controlled Open Quantum System. *Phys. Rev. X* **7**, 011034 (2017).

[23] Klein, A., Bruderer, M., Clark, S.R., & Jaksch, D. Dynamics, dephasing and clustering of impurity atoms in Bose–Einstein condensates. *New J. Phys.* **9,** 411 (2007).

[24] Buča B., & Prosen. T. A note on symmetry reductions of the Lindblad equation: transport in constrained open spin chains. *New J. Phys.* **14**, 073007 (2012).

[25] Yu, S. SO(4) Symmetry and Off-Diagonal Long-Range Order in the Hubbard Bilayer. *Commun. Theor. Phys.* **28**, 23 (1997).

[26] Yang, C.N. η pairing and off-diagonal long-range order in a Hubbard model. *Phys. Rev. Lett.* **63**, 2144 (1989).

[27] Lostaglio, M., Jennings, D., & Rudolph, T. Description of quantum coherence in thermodynamic processes requires constraints beyond free energy. *Nat. Comm.* **6**, 6383 (2015).

[28] El-Ganainy, R., et al. Non-Hermitian physics and PT symmetry. *Nat. Phys.* **14**, 11 (2018).

[29] Macieszczak, K., et al. Towards a theory of metastability in open quantum dynamics. *Phys. Rev. Lett.* **116,** 240404 (2016).

[30] Wilczek F., Quantum time crystals. *Phys. Rev. Lett.* **109**, 160401 (2012).

[31] Else, D.V., Bauer, B., & Nayak, C. Floquet Time Crystals. *Phys. Rev. Lett.* **117**, 090402 (2016).

[32] Iemini, F., et al. Boundary time crystals. *Phys. Rev. Lett. **121,*** 035301 (2018).

[33] Essler, F.H.L., Frahm, H., Göhmann, F., Klümper, A., & Korepin, V. E. *The One-Dimensional Hubbard Model* (Cambridge University Press, Cambridge, 2005).

[34] Daley A. Quantum Trajectories and open many-body quantum systems, Adv. Phys. 63, 77 (2014).

[35] Al-Assam A., Clark. S. R. & Jaksch D. Tensor network theory-Part 1: Overview of core library and tensor operations, J. Stat. Mech. 093102 (2017).

**Acknowledgements**

**Author Contributions**

B.B identified the conditions for non-stationarity and carried out the analytical calculations. J.T. performed the numerical calculations. D.J., B.B and J.T. discussed and analysed the results and their physical implications and contributed significantly to writing the manuscript.

**Competing Interests**

The authors declare no competing interests.

**Figures:**

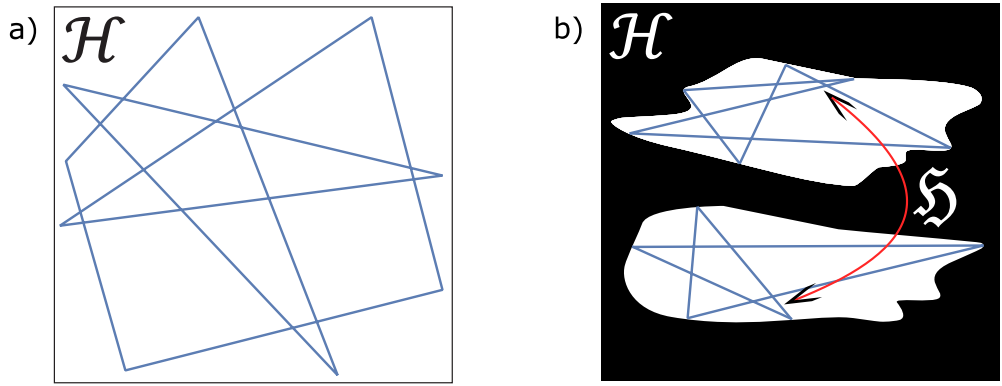

*Figure 1: The Hilbert space $\mathcal{H}$ separated into a decaying part (black) and a non-decaying part (white). a) Ergodic time-evolution, indicated by a blue trajectory, explores the entire connected non-decaying space thus leading to stationarity of observables after a transient period. b) Dissipation may split the non-decaying part into disjoint sectors. The dark Hamiltonian $\mathfrak{H}$ drives transitions between them. Observables that are not entirely contained in one of these parts show continued oscillations and thus non-ergodic behaviour after a transient period.*

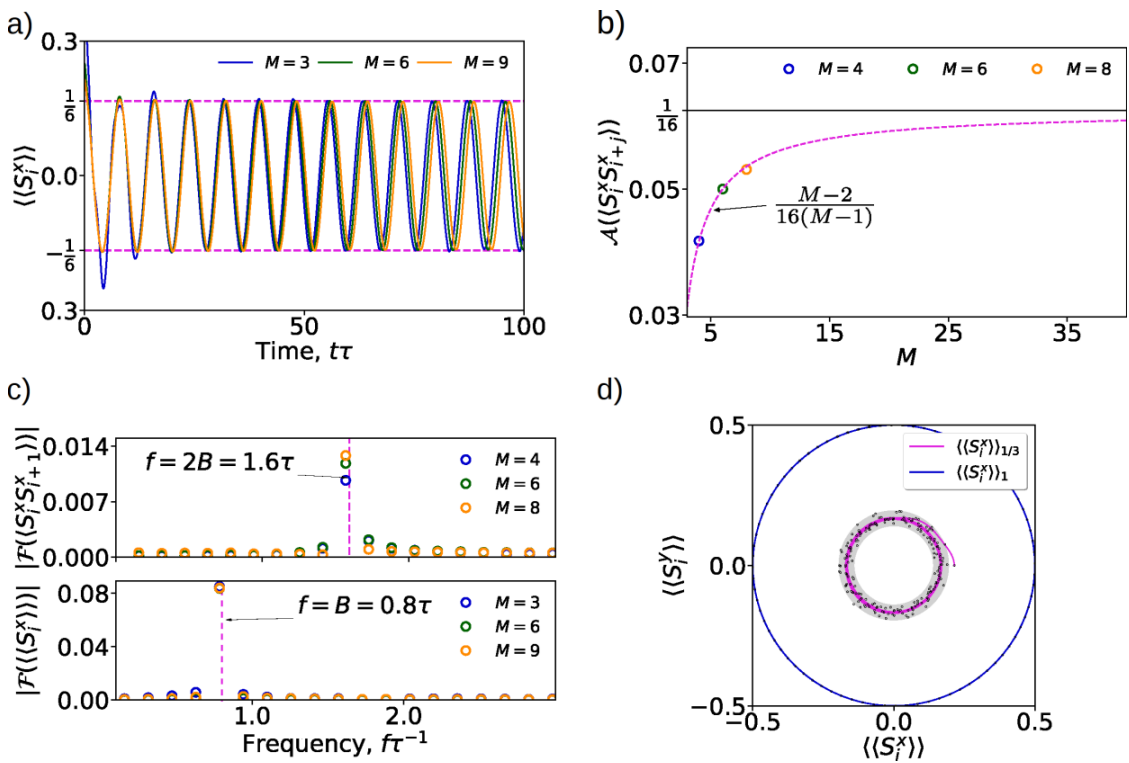

Figure 2: Dynamics of spin observables in the number-dephased Hubbard model following a quench. a) The bulk-averaged fermion spin $\langle\langle S_i^x(t)\rangle\rangle$ for various system sizes $M$. The evolution starts from the half-filled lattice state $\langle\langle S_i^x\rangle\rangle_{1/3}$ without double-occupancies where every third fermion is polarized along $-x$ while all others are polarized along $x$. A long-time amplitude of $1/6$ independent of $M$ is obtained analytically. b) Amplitude of the oscillations of $\langle S_i^x(t)S_{i+j}^x(t)\rangle$ in the long-time limit for different $M$ and arbitrary $i,j$. The starting state is the maximally polarized quarter filled state with a fermion put on every second site. The magenta dashed curve shows the analytical result converging to $1/16$ in the limit $M \to \infty$. c) The spectra obtained from the dynamical evolution in a) and b) for times $t \in [20,100]/\tau$ are strongly peaked around multiples of $B$ as expected in the long-time limit. d) Traces of the polarization in the $xy$-plane starting from the maximally polarized starting state $\langle\langle S_i^x\rangle\rangle_1$(blue curve) and from the state $\langle\langle S_i^x\rangle\rangle_{1/3}$ (magenta curve) and $M = 9$. The solid lines are averages over 2000 trajectories (see Methods). The markers show values from a single realization and the shaded area indicates the range of typical fluctuations of a realization. All calculations were carried out for $B = 0.8\tau$, $U = \tau$ and $\gamma = 0.4\sqrt{\tau}$ without disorder $\epsilon_j = 0$.

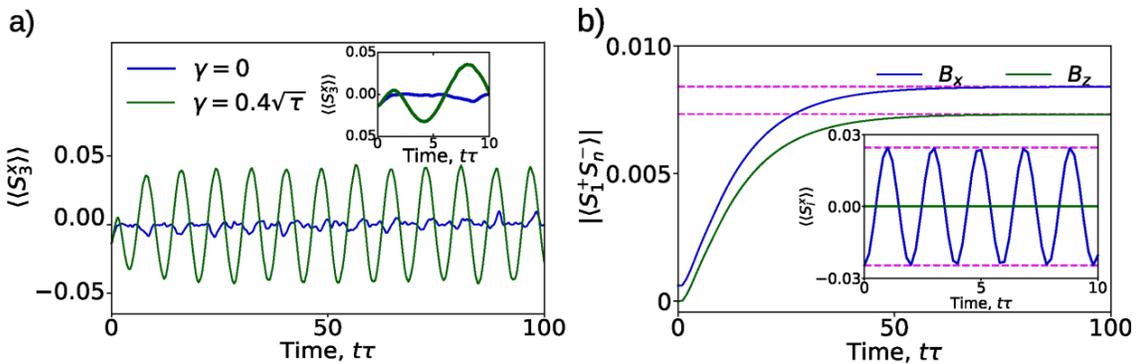

Figure 3: Time evolution of spin observables in both the open and closed Hubbard model during a quench. a) Dynamics of $\langle S_3^x(t)\rangle$ for the open and the closed system following a quench (the inset shows relaxation and the first revival of the finite-size closed system). The initial state is the Hubbard model ground state with $N = 5$ particles, $M = 7$ sites, $U = \sqrt{2}\tau$, $B = 0$ and no disorder. At time $t = 0$ the system is quenched to $U = \tau$, $B =$

$0.8\tau$ and disorder $\varepsilon_j \in [0.0, 0.16]\tau$. b) Dynamics of $|\langle S_1^+(t) S_n^-(t)\rangle|$ starting from a high temperature thermal state $\propto \exp(-\beta H)$ with $\beta = 0.2/\tau$, $U = \tau$ and $B = 0.2\tau$ of $N = 4$ fermions in $M = 5$ sites. The magnetic field initially points either along the x-direction (blue lines) or the z-direction (green lines). At time $t = 0$ dephasing $\gamma = 0.4\sqrt{\tau}$ is switched on and $B$ pointed along the z-direction. Long-range correlations emerge from the initially thermal state through dephasing. The inset shows $\langle\langle S_i^x(t)\rangle\rangle$ and the magenta dashed lines represent analytical values in the long-time limit.

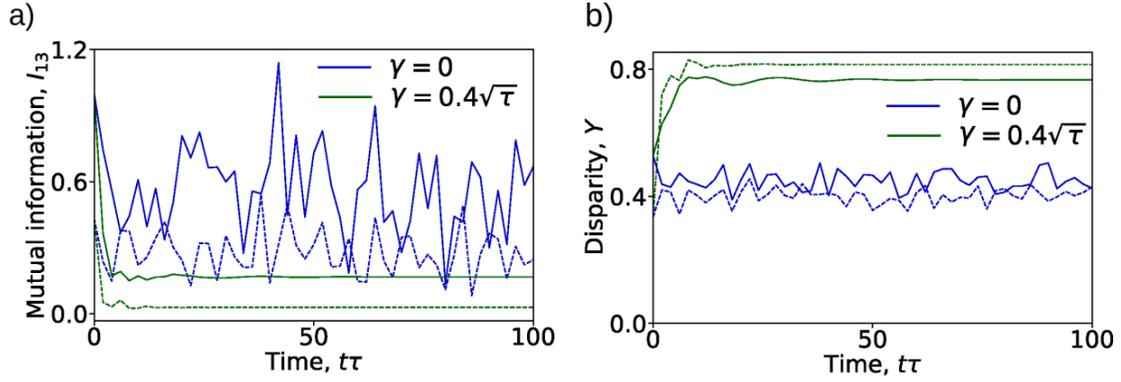

Figure 4: Measures of complexity applied to the evolution of a random initial state of the open and closed Hubbard model. The data is for $N = 3$ particles on $M = 4$ sites (dotted lines) and $N = 4$ particles on $M = 4$ sites with the first, second and fourth sites being polarized along the x-direction and the third site along $-x$ (solid lines). a) shows the quantum mutual information between sites 1 and 3 that decreases quickly during the initial part of the evolution and is noticeably less for the open system case. b) shows the average disparity $Y$ between lattice sites growing with time for the open system dynamics and remaining approximately constant in the closed system. Other parameters as in Fig. 2.


\begin{titlepage}
	\maketitle
\end{titlepage}
	\section*{Supplementary Methods}
	\subsection*{Introduction}
	Let us recall the Lindblad equation,
	\begin{eqnarray}
	&&\frac{\dd}{\dd t}\rho(t) = \LL\rho(t):= \nonumber \\
	&&-\ii [H,\rho(t)] + \sum_\mu  \left(2L_\mu \rho(t) L_\mu^\dagger - \{L_\mu^\dagger L_\mu,\rho(t)\}\right),
	\label{eq:lindeq}
	\end{eqnarray}
	and the eigenvalue equation for the Liouvillian,
	\begin{equation}
	\LL \rho = \lambda \rho  \iff  \LL \rho^\dagger = \lambda^* \rho^\dagger. \label{eigenmodes}
	\end{equation} 
	We set $\hbar=1$ for simplicity of notation. We are concerned with cases where $\lambda$ is purely imaginary. The corresponding eigenmodes are called oscillating coherences \cite{AlbertJiang1}, or limit cycles \cite{Bellomo}. Here we emphasize that our results may be also understood as a framework for constructing genuine \emph{many-body} quantum synchronization \cite{Bellomo, Lee, Xu}. 
	
		\subsection*{Formal classification of the asymptotic subspaces and relation to decoherence-free subpaces}
	
	To make connection with the existing literature we will discuss the formal classification of the asymptotic subspace of the Liouvillian that we constructed here. 
	
	We will follow the terminology of Supplementary Reference~\citen{AlbertJiang2} (see also Supplementary Referece~\citen{BN}). The asymptotic subspaces of the quantum Liouvillian from the main text, that the dynamics is leads to in the long-time limit, is a \emph{multi-block} structure, which is the most general possible form that it can have \cite{AlbertJiang2}. 
	
	More specifically, for the Hubbard model example in the main text, by indexing the subspaces of $S^+ S^-$ ($N$) as $\mu$ ($\nu$) (where $S^\alpha$ are the total spin operators and $N$ is the total particle number) the basis of the asymptotic subspace can be written in the form,
	\begin{equation}
	\ket{z^{\mu,\nu}_1}\bra{z^{\mu,\nu}_2} \otimes P^{S^+S^-}_\mu P^{N}_\nu,
	\end{equation} 
	where $P^{S^+S^-}_\mu$ and $P^{N}_\nu$ are projectors to the corresponding subspaces of $S^+ S^-$ and $N$, respectively, and $\ket{z^{\mu,\nu}_{1,2}}$ are the corresponding eigenstates of $S^z$. 
	
	The form of the multi-block of the other examples featuring the Hubbard model, discussed later in in this supplementary, are more difficult to construct, due to the fact that the corresponding quantum Liouvillian will not be unital (i.e. $\LL \one \neq0$).  
	
	We now further remark on the differences between this type of structure and a dynamical decoherence-free subspace widely studied in literature (for references beyond the ones cited in the main text see e.g. Supplementary References \citen{Karasik, Brooke,Wu}). A decoherence-free subspace is a subspace of the Hilbert space invisible to dissipation, i.e. ${\cal H}_{DFS} \subseteq {\cal H}$ and for any pure state $\ket{\psi(t)} \in {\cal H}_{DFS}$ (with $\rho(t)=\ket{\psi(t)}\bra{\psi(t)}$) we have $\partial_t \tr{\rho(t)^2}=0$ (e.g. Def. 2 of Supplementary Reference \citen{Karasik}).  In other words, all pure states in the subspace undergo coherent time evolution given by the system's Hamiltonian and remain pure \cite{AlbertJiang2}. Therefore, they may be understood as restriction of the closed system's Hamiltonian to a subspace to which the dissipative time evolution guides the system in the long time limit. For instance, the open XXZ spin ring later in this supplementary material is a decoherence-free subspace. This open XXZ spin ring is, to our knowledge, the first example of a decoherence-free subspace in an open quantum \emph{many-body} system. 
	
	In contrast to a decoherence-free subspace the multi-block structure, for which we provide sufficient criteria, is affected by the dissipation. The asymptotic dynamics is coherent, but it consists of generally mixed states. The physical properties of these mixed states are affected by the dissipation. For instance, the dissipation may induce off-diagonal long-range order (like in the example in the main text), or currents of some quantity (like in the examples in the later sections). Furthermore, when we study the quantum stochastic process \cite{Daley} corresponding to the master equation given in Supplementary Equation \eqref{eq:lindeq}, the dynamics in a decoherence-free subspace is purely deterministic (every quantum trajectory is the same, like in the closed system), whereas in the multi-block structure the dynamics is stochastic (the ensemble of trajectories is non-trivial). 
	
		\subsection*{Generalizing beyond the Markovian framework}
	
	It is fairly straightforward to generalize the discussion in the main text beyond the Markovian framework. Let the full Hamiltonian of the system (with Hilbert space ${\cal H}_S$) and bath (with Hilbert space ${\cal H}_B$) be given as
	\begin{equation}
	H= H_S+H_{SB}+H_B,
	\end{equation}
	where $H_S \in {\cal H}_S$ is the system's Hamiltonian, $H_B \in {\cal H}_B$ is the Hamiltonian of the bath and $H_{SB} \in {\cal H}_S \otimes {\cal H}_B$ is the system-bath interaction. In general, we may write $H_{SB}=\sum_j S_j \otimes B_j$, where $S \in {\cal H}_S$ and $B\in {\cal H}_B$. 
	
	Let $A$ be an eigenoperator of the full Hamiltonian, $[A,H]=\lambda A$. As $H$ is Hermitian $\lambda \in \RaR$. The case studied in the main text would correspond to $A=A_S \otimes \one_B$, with $A_S \in {\cal H}_S$ and,
	\begin{equation}
	[A_S,H_S]=\lambda A_S, \qquad [A_S, H_{SB}]=0. \label{gencond}
	\end{equation}
	The full time evolution of the open system is given as,
	\begin{equation}
	\rho_S(t):={\cal \hat{T}}_t \rho(0)=\tr_B \left[ e^{-\ii H t} \rho(0) e^{\ii H t} \right],
	\end{equation}
	where $\tr_B$ represents tracing over the bath degrees of freedom. 
	
	Let $\rho_\infty \in {\cal H}_S$ be a stationary state of the dynamical map in the sense that for some $\rho'(0)$ we have,
	\begin{equation}
	\rho_\infty=\lim_{t \to \infty}{\cal \hat{T}}_t \rho'(0). \label{asymp}
	\end{equation} 
	Such as state always exists (e.g. any incoherent mixture of the eigenstate of $H$ will produce such a state trivially). We look at the asymptotic time evolution of $\rho_{nm}=A^n \rho'(0) (A^\dagger)^m$. It follows from $[H,A]=\lambda A$ that,
	\begin{equation}
	e^{-\ii H t} A e^{\ii H t}= e^{\ii \lambda t} A. \label{heisenberg}
	\end{equation}
	We decompose the density matrix of the full system and bath as $\rho(t)=\sum_{i,j} s_i(t) \otimes b_j(t)$ and write,
	\begin{equation}
	\rho_S(t)=\tr_B \sum_{i,j} s_i(t) \otimes b_j(t)=\sum_{i,j}c_j(t) s_i(t), \label{decomp}
	\end{equation}
	where $s_j(t) \in {\cal H}_S$ and  $b_j(t) \in {\cal H}_b$ and $c_j(t)=\tr_B b_j(t)$.
	
	Then, 
	\begin{eqnarray}
	\lim_{t \to \infty} {\cal \hat{T}}_t \rho_{nm}&=&\lim_{t \to \infty} \tr_B \left[ e^{-\ii H t}A^n \rho'(0) (A^\dagger)^m  e^{\ii H t} \right] \nonumber \\
	&=&\lim_{t \to \infty}  \tr_B \left[e^{\ii n \lambda t} A^n e^{-\ii H t} \rho'(0)   e^{\ii H t}(A^\dagger)^m  e^{- \ii m \lambda t} \right]  \nonumber \\
	&=&\lim_{t \to \infty}  \tr_B \left[e^{\ii n \lambda t} A^n \left( \sum_{i,j} s_i(t) \otimes b_j(t) \right) (A^\dagger)^m  e^{- \ii m \lambda t} \right]  \nonumber \\
	&=&\lim_{t \to \infty}  e^{\ii n \lambda t}    A^n \left( \sum_{i,j} c_j(t) s_i(t)   \right) (A^\dagger)^m e^{- \ii m \lambda t}  \nonumber \\
	&=& \lim_{t \to \infty} e^{\ii n \lambda t}  A^n \rho_\infty (A^\dagger)^m e^{- \ii m \lambda t},  
	\end{eqnarray}
	where to obtain the second equality we repeatedly inserted $\one= e^{-\ii H t}  e^{\ii H t} $ between the products of $A$'s and used Supplementary Equation \eqref{heisenberg}. To obtain the third and fourth equality we used Supplementary Equation \eqref{decomp} and the fact that $A \in {\cal H}_S$. The last equality was obtained by recognizing that $\rho_\infty = \lim_{t \to \infty} \rho_S(t)$ is a stationary state.
	
	We emphasize that $\rho_{nm}$ are \emph{not} density matrices, but their linear combinations can be chosen to be.

	\subsection*{Dynamical decoherence-free many-body subspaces}
	
	We will define what we mean by a \emph{standard} dark Hamiltonian $\dH:=[\tilde{H},\bullet]$ (or dynamical many-body decoherence-free subspace) with $\tilde{H}=\tilde{H}^\dagger$. We want the spectrum of $\dH$ to be made up of \emph{pure} mutually orthogonal eigenstates and to have the same left and right eigenvectors. 
	Namely:
	\begin{enumerate}
		\item We first wish that, \begin{equation}
		\dH \ketbra{\phi_n}{\phi_m}=(\omega_n-\omega_m)\ketbra{\phi_n}{\phi_m}, \quad \omega_k \in \RaR, \forall k. \label{eigs1}
		\end{equation} \label{cond1}
		\item Where we also desire that eigenmodes $\rho_{nm}:=\ketbra{\phi_n}{\phi_m}$ are mutually orthogonal in the Hilbert-Schmidt sense, $\tr \rho_{nm}^\dagger \rho_{n' m'}=\delta_{n,n'} \delta_{m,m'}$. \label{cond2}
		\item Finally we require that the left and right eigenvectors match ensuring that $\dH=\dH^\dagger$. \label{cond3}
	\end{enumerate}
	
	We state the details of the conditions under which properties~\ref{cond1}-\ref{cond3} will be fulfilled in Supplementary Theorem~\ref{th1}.
	
	\begin{theorem} \label{th1}
		A set of mutually orthogonal vectors $\{\ket{\phi_1},\ket{\phi_2}, \ldots \}$ forms a set of eigenvectors of a standard dark Hamiltonian (or a decoherence-free subspace) in the sense of properties~\ref{cond1}-\ref{cond3} iff the following conditions are fulfilled,
		\begin{enumerate}[label=(\alph*)]
			\item $
			\left(\ii H +  \sum_k \gamma_k L_k^\dagger L_k \right) \ket{\phi_n} = \lambda_n \ket{\phi_n}, \quad \forall n, \label{cond1th1}
			$
			\item$
			L_k \ket{\phi_n}=\lambda_{k,n}  \ket{\phi_n} ,
			$ and $\sum_k \gamma_k \abs{\lambda_{k,n}}^2=\Re{\lambda_n},  \quad \forall n $, \label{cond2th1}
			\item $\Re \left[ \sum_k  \gamma_k \left ( 2\lambda_ {k,n} \lambda^*_{k,m} -\abs{\lambda_{k,n}}^2-\abs{\lambda_{k,m}}^2 \right ) \right]=0, \forall n,m $. \label{cond3th1}
		\end{enumerate} 
		Then the eigenvalues from Supplementary Equation \eqref{eigs1} are given as \begin{equation}
		\omega_n-\omega_m=\im\left(\sum_k 2  \gamma_k \lambda_ {k,n} \lambda^*_{k,m}  - \bra{\phi_m}H \ket{\phi_n}\right). \label{eqth1}
		\end{equation}
	\end{theorem}
	\begin{proof}
		Conditions~\ref{cond1th1} and~\ref{cond2th1} are just the well-known conditions of Theorem 1 of Supplementary Reference \citen{Kraus} for pure stationary states $\LL(\ket{\phi_n}\bra{\phi_n})=0$. Condition~\ref{cond3th1} can be shown as follows: We begin by writing out Supplementary Equation \eqref{eigenmodes} for $\rho=\ket{\phi_m} \bra{\phi_n}$, and taking the product with $\bra{\phi_m}$ from the left and $\ket{\phi_n}$ from the right. Then we demand that the corresponding eigenvalue $\lambda$ is purely imaginary and use conditions~\ref{cond1th1} and~\ref{cond2th1} of the theorem. Recalling the definitions in Supplementary Equation \eqref{eigs1}, this also leads to Supplementary Equation \eqref{eqth1}.
	\end{proof}
	
	\begin{theorem} \label{th2}
		If there is no subspace ${\cal S} \subset {\cal H}$ that is orthogonal to the space of dark states\footnote{Dark states $\ket{\phi_n}$ are defined as $L_k \ket{\phi_n}=0$ and $H \ket{\phi_n}=\omega_n \ket{\phi_n}$, $\forall k$.} ${\cal D}$ (i.e., ${\cal S} \perp {\cal D}$) such that $L_k {\cal S} \subset {\cal S}$, then the only oscillating coherences are vectors in a decoherence free subspace in the sense that they have the form $\ket{\phi_n}\bra{\phi_m}$ with $\ket{\phi_i} \in {\cal D}$ being dark states. 
	\end{theorem}
	\begin{proof}
		Theorem 2 of Supplementary Reference \citen{Kraus} guarantees that if there is no such subspace ${\cal S}$ then the only stationary states are dark states $\LL \ket{\phi_n}\bra{\phi_n}=0$. According to Theorem 18 of Supplementary Reference \citen{BN} all oscillating coherences must be of the form $A \rho_\infty$ where $\LL \rho_\infty=0$ and $A$ is an operator. We know that the only stationary states are of the form $\ket{\phi_n}\bra{\phi_n}$. Writing $\ket{\Psi}=A \ket{\phi_n}$, we have that any oscillating coherence must satisfy 
		$-\ii [H,\ket{\Psi}\bra{\phi_n}] + \sum_\mu \gamma_\mu \left(2L_\mu \ket{\Psi}\bra{\phi_n} L_\mu^\dagger - \{L_\mu^\dagger L_\mu,\ket{\Psi}\bra{\phi_n}\}\right)= \ii \Lambda \ket{\Psi}\bra{\phi_n}$, with $\Lambda \in \RaR$. Using $L_k \ket{\phi_n}=0$ and $H \ket{\phi_n}=\omega_n \ket{\phi_n}$ we are left with $-\ii H \ket{\Psi}\bra{\phi_n} -L_\mu^\dagger L_\mu\ket{\Psi}\bra{\phi_n}=\ii \Lambda \ket{\Psi}\bra{\phi_n} $. We take the product of this with $\bra{\Psi}$ from the left and $\ket{\phi_n}$ from the right, thus we have $\bra{\Psi}L^\dagger_k L_k \ket{\Psi} = \abs{\abs{L_k \ket{\Psi}}}^2 \ge 0$. We then use the fact that $H$ is Hermitian and so $\bra{\Psi} H \ket{\Psi} \in \RaR$. Then the only way to satisfy the eigenequation is iff $L_k \ket{\Psi}=0$ and $H \ket{\Psi}=\Lambda \ket{\Psi}$, therefore $\ket{\Psi} \in {\cal D}$ is also a dark state.
	\end{proof}

	\subsection*{Theorems on general complex coherent dynamics under dissipation}
	In this section we go beyond $\dH$ being a standard Hamiltonian, i.e. $\dH \neq [\tilde{H},\bullet]$, and move to cases when the eigenmodes of $\dH$ are not pure states and when $\dH$ is not Hermitian. \\
	\vspace{40pt}
	
	\begin{theorem} \label{th3}
		Let $\LL \rho_\infty=0$. If there exists a non-trivial operator $A$ with the property 
		\begin{equation}
		[H,A] \rho_\infty=\lambda A  \rho_\infty, \quad [L_k,A]\rho_\infty=[L^\dagger_k,A]L_k\rho_\infty=0, \quad \forall k, \label{eqth3}
		\end{equation}
		then the state $\rho=A \rho_\infty $ is an eigenstate of the Liouvillian with purely imaginary eigenvalue, $\LL \rho=\ii \lambda \rho, \lambda \in \RaR$. 
	\end{theorem}
	
	\begin{proof}
		Take the conjugate transpose of Supplementary Equation \eqref{eqth3}. We get $\rho_\infty [H,A^\dagger] =-\lambda^* \rho_\infty A^\dagger$ and $\rho_\infty[L_k^\dagger,A^\dagger]=\rho_\infty L_k^\dagger[L_k,A^\dagger]=0, \forall k$, as $\rho_\infty$ is Hermitian. Now define a superoperator $\hat{S} \rho:=[H,A^\dagger A] \rho, \forall \rho \in \End({\cal H})$. The superoperator is a commutator of two Hermitian operators and is therefore clearly skew-Hermitian (with purely imaginary eigenvalues). We take the Hilbert-Schmidt inner product of this superoperator with $\rho_\infty$. The result is purely imaginary, $\tr\left( \rho_\infty [H,A^\dagger A] \rho_\infty \right)=\ii \theta, \theta \in \RaR$, as $\hat{S}$ is skew-Hermitian. It is straightforward to calculate, using Supplementary Equation \eqref{eqth3} and the conjugate transpose of Supplementary Equation \eqref{eqth3}, that $\tr\left( \rho_\infty \hat{S} \rho_\infty\right)=(\lambda-\lambda^*) \tr\left( \rho_\infty A^\dagger A \rho_\infty\right)$. However, $\tr\left( \rho_\infty A^\dagger A \rho_\infty\right)= \abs{\abs {A\rho_\infty}}>0$. Thus the rhs is purely real, whereas the lhs is purely imaginary. Therefore, $\lambda=\lambda^*$ (i.e, $\lambda \in \RaR$).
		
		We define a new superoperator $\hat{A} \rho:=A \rho, \forall \rho \in \End({\cal H})$. Using the definition of the Liouvillain $\LL$ from Supplementary Equation \eqref{eq:lindeq} and using Supplementary Equation \eqref{eqth3} (and its conjugate transpose), it is straightforward to show that $[\LL,\hat{A}]\rho_\infty=\ii \lambda \hat{A}\rho_\infty$. Then, using $\LL \rho_\infty=0$, the statement of the theorem directly follows. 
	\end{proof}

	\begin{corollary}
		\label{cor1}
		In particular, if 
		\begin{equation}
		[H,A] =\lambda A,  \quad [L_k,A]=[L^\dagger_k,A]=0, \quad \forall k, \label{eqco1}
		\end{equation}
		then $\LL \rho_{nm}=\ii \lambda (n-m) \rho_{nm}, \lambda \in \RaR$, with $\rho_{nm}=A^n \rho_\infty (A^\dagger)^m.$
	\end{corollary}
	\begin{proof}
		From Supplementary Equation \eqref{eqco1} it also follows that $[H,A^n] =n \lambda A^n , [L_k,A^n]=[L^\dagger_k,A^n]=0, \forall k,n $. Taking the conjugate transpose of that we have $[H,(A^\dagger)^n] =-n \lambda (A^\dagger)^n$, $[L_k,(A^\dagger)^n]=[L^\dagger_k,(A^\dagger)^n]=0, \forall k,n $. Assume without loss of generality that $n>m$: then it follows from Supplementary Equation \eqref{eqco1}, its conjugate transpose and Supplementary Equation \eqref{eq:lindeq} that $\LL A^m \rho_\infty (A^\dagger)^m=0$, i.e. $\rho'_\infty= A^m \rho_\infty (A^\dagger)^m$ is also a stationary state. Define $A'=A^{n-m}$, we then have $\rho_{nm}=A' \rho'_\infty$ and we can now simply invoke Supplementary Theorem~\ref{th3} and arrive at the main statement of the corollary. 
		Recalling that it also directly follows that the state $\rho_{nm}^\dagger$ is an eigenstate with an eigenvalue of opposite sign  $\LL \rho_{nm}^\dagger=-\ii(n-m) \lambda \rho_{nm}^\dagger, \lambda \in \RaR$, we can repeat the same procedure assuming that $m>n$. This proves the corollary. 
	\end{proof}

 	\section*{Supplementary Discussions}
 	We now proceed to give examples and discuss the proposed experimental implementations. 
 	\subsection*{Dynamical decoherence-free subspace: The XXZ spin ring }
 	
 	We study the $n$-site Heisenberg XXZ spin chain with periodic boundary conditions (spin ring),
 	\begin{equation}
 	H_{\rm{XXZ}}=\sum_{j=1}^{n} \sigma^+_j \sigma^-_{j+1}+ \sigma^-_j \sigma^+_{j+1} +\Delta  \sigma^\z_j \sigma^\z_{j+1},
 	\end{equation}
 	where $\Delta$ is an anisotropy parameter and the spin-$1/2$ operators on site $j$ are $\sigma^\alpha_j=\one_2^{\otimes (j-1)} \otimes \sigma^\alpha \otimes\one_2^{\otimes (n-j)}$ (with $\sigma^\alpha$ being the standard Pauli matrices and $\one_2$ is the identity matrix of size $2\rm{x}2$). For the sake of simplicity, we study the dimensionless version of the model. We also introduce a single ultra-local loss term $L= \gamma \sigma_1^-$ with loss rate $\gamma \ge 0$. In the long-time limit this setup induces a dynamical decoherence-free many-body subspace (see Supplementary Theorems~\ref{th1} and~\ref{th2}) that is formed from pure states that are eigenstates of the Hamiltonian $H_{\rm XXZ} \ket{\phi_n}=\omega_n \ket{\phi_n}$ and that are annihilated by $L \ket{\phi_n}=0$. Thus, $\LL \ketbra{\phi_n}{\phi_m}= \ii (\omega_n-\omega_m) \ketbra{\phi_n}{\phi_m}$, and any state of the form $\ket{\Psi}=\ket{\phi_n} +e^{\ii \alpha} \ket{\phi_m}$ will undergo oscillations when $\omega_n \neq \omega_m$, $\forall \alpha$. 
 	
 	The Liouvillian of the model was studied through exact diagonalization, and by requiring that the Bethe ansatz solutions for the eigenvectors of $H_{\rm{XXZ}}$ are dark states of $L$. The spectrum of the corresponding Liouvillian is shown in Supplementary Figure~\ref{fig:XXZ} for different system sizes $n$. We find that the number of distinct purely imaginary eigenvalues scales sub-quadratically with $n$. The spectrum in Supplementary Figure~\ref{fig:XXZ} also shows the formation of oscillatory patterns in the spectral densities that hint at the formation of Bethe strings in the values of the complex quasi-rapidities of the related integrable model. The eigenstates with purely imaginary eigenvalues are states with nodes on the loss site (and may also be understood as lattice scarring \cite{Jordi}, but a 1D quantum many-body version \cite{Turner}). The purely imaginary eigenvalues are not present at the non-interacting point of the model ($\Delta=0$).
 	
 	In the thermodynamic limit the spectrum of the lossy XXZ spin ring Liouvillian seems to be incommensurate. This may again lead to relaxation, though it would likely require much longer time than in the corresponding closed XXZ ring
 	
 	\begin{suppfigure}[!]
 		\begin{center}
 			\vspace{-1mm}
 			\includegraphics[width=1\textwidth]{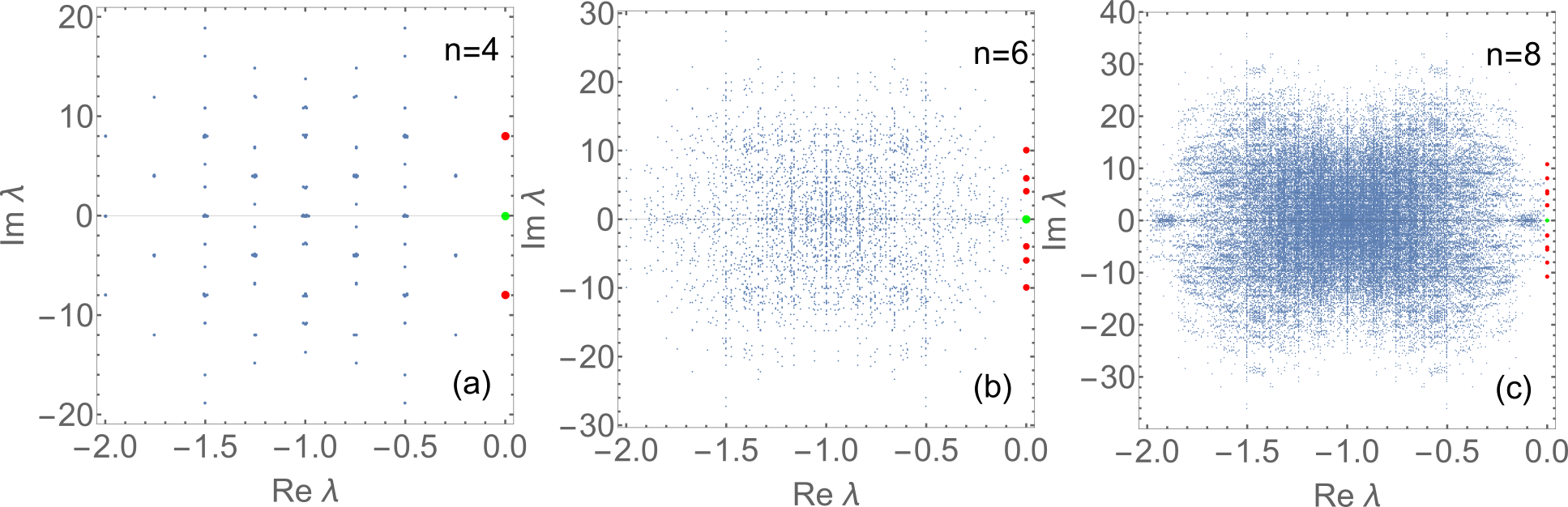}
 			\vspace{-4mm}
 		\end{center}
 		\caption{The spectrum of the XXZ Liouvillian with a single loss term. We study three system sizes n=4 (a), n=6 (b) and n=8 (c) in the gapped regime with $\Delta=2$ and $\gamma=1$ and plot the real and imaginary part of the eigenvalues of $\LL$. The larger red points indicate the purely imaginary eigenvalues and the green points the stationary states (with eigenvalue 0).}    
 		\label{fig:XXZ}
 	\end{suppfigure} 
 	
 	In the long-time limit this setup restricts the dynamics to the dynamical decoherence-free many-body subspace that undergoes continued oscillations. We take $\Delta=1.1$, $n=4$, $\gamma=1$ and show the dynamics in Supplementary Figure~\ref{fig:sup2} for $O(t)=\ave{\sigma^\x_2(t)}$. The time evolution is initialized in a random initial state. Numerical investigation indicates that the dynamical decoherence-free subspace is robust to terms in the Hamiltonian that break integrability, but it is not robust to the presence of additional Lindblad operators. In particular, adding a second loss term destroys the dynamical decoherence-free subspace. We thus compare the time evolution of the same observable starting in the same initial state and with the same parameters, but with two loss terms $L_1= \gamma_1 \sigma_1^-$, $L_2= \gamma_2 \sigma_2^-$ ($\gamma_1=\gamma_2=\gamma$) in Supplementary Figure~\ref{fig:sup2}. This case clearly shows relaxation to stationarity. Furthermore, the stationary state is unique.

 	\begin{suppfigure}[!]
 		\begin{center}
 			\vspace{-1mm}
 			\includegraphics[width=1\textwidth]{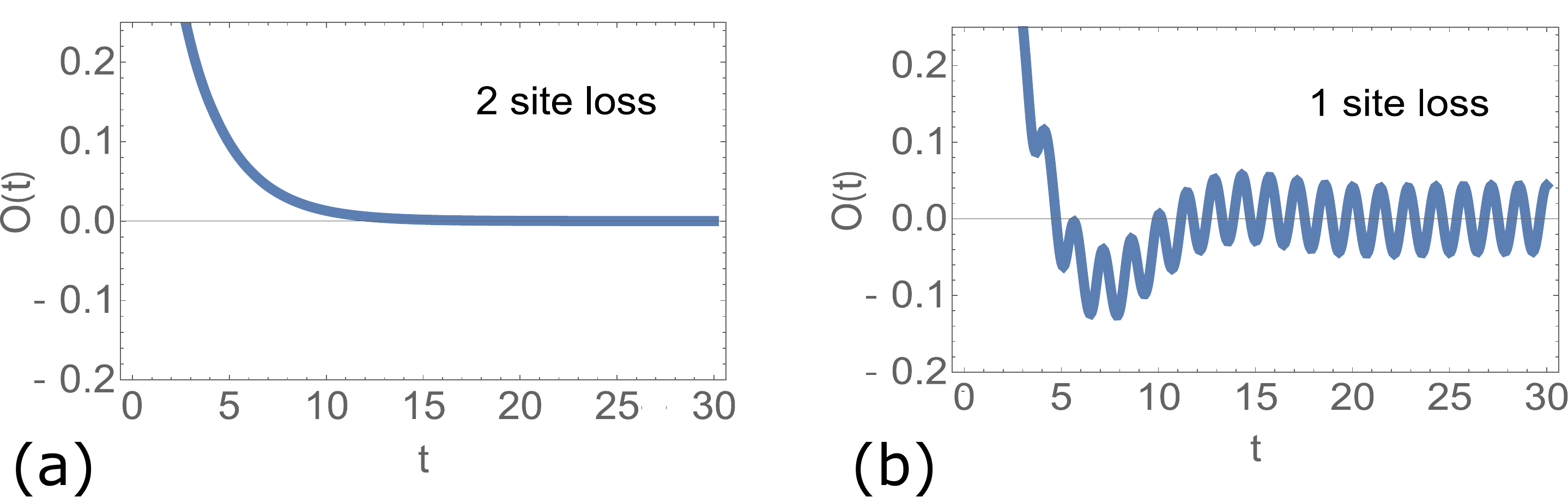}
 			\vspace{-4mm}
 		\end{center}
 		\caption{Comparison of the dynamics of the XXZ spin ring for $\Delta=1.1$, $n=4$, $\gamma=1$ with two (a) and one (b) loss terms. Introducing two loss terms destroys the dynamical decoherence-free subspace and leads to relaxation to stationarity (a). One loss term, instead, preserves this dynamical subspace and leads to persistent oscillations. Results shown for $O(t)=\ave{\sigma^\x_2(t)}$. }    
 		\label{fig:sup2}
 	\end{suppfigure}

 	%
 	
	\subsection*{Multi-block structures: Additional examples with the Hubbard model}
	
	\subsubsection*{Spin dephasing}
	Due to the spin and $\eta$-pairing symmetries of the $D$-dimensional Hubbard model the conditions of Supplementary Corollary~\ref{cor1} will be fulfilled if we take any set of Lindblad operators contained exclusively in either symmetry sector. This will construct a dark Hamiltonian. For instance, in contrast to the main text, we may take a set of ultra-local spin dephasing Lindblad operators $L_k=\gamma_k S^\z_k$ (and set a constant lattice potential $\epsilon_j$=0). The difference to the example in the main text is that the stationary state is now given as, 
	\begin{equation}
	\rho_\infty=C \exp\left(\beta_0\eta^\z+ \beta_1(\eta^+ \eta^-)+\beta_2 S^\z\right), \label{ESS}
	\end{equation}
	where $\beta_m$ play the role of chemical potentials of the grand canonical ensemble, and $C$ is a normalization constant. This dissipation also constructs a dark Hamiltonian with the following eigenmodes,
	\begin{equation}
	\dH \left((\eta^+)^m \rho_\infty  (\eta^-)^n\right)=2(m-n) \mu \left((\eta^+)^m \rho_\infty  (\eta^-)^n\right),  \label{hubbarddark}
	\end{equation}
	where $\mu$ is the chemical potential. Thus starting from an initial state with a well-defined number of particles and doublons we arrive at the stationary state given in Supplementary Equation \eqref{ESS}. However, if we start from an initial state that is off-diagonal in both particle number and the number of doublons, we will get oscillations of the form of Supplementary Equation \eqref{hubbarddark}. The eigenmodes of the dark Hamiltonian all have $\eta$-pairing symmetry, which is of particular interest when studying superfluidity and superconductivity as $\eta$-paired states have long-range off-diagonal order \cite{Yang} for all dimensions $D$. Thus all the eigenmodes of the dark Hamiltonian are superconductive \cite{Yang}. 
	
	\subsubsection*{Transport of doublons}
	
	Let us now take Lindblad jump operators that drive doublons, i.e. in one spatial dimension $L_1=\gamma_1 \eta^+_1$ and $L_2=\gamma_2 \eta^-_n$, and an analogously chosen pair in higher spatial dimension. 
	
	For the same reasons as in the example in the main text we obtain a dark Hamiltonian. The crucial difference now is that the stationary state $\rho_\infty$ will be different. As we are driving doublons, the eigenstates of the dark Hamiltonian will be \emph{non-equilibrium} mixed states, which support a doublon current in addition to having coherent oscillations in the spin sector.

\subsection*{Experimental implementation with ultracold atoms}

Ultracold fermionic atoms trapped in an optical lattice provide a clean experimental realization of the Hubbard model \cite{Esslinger}. The dynamics of the atoms can be restricted to two atomic hyperfine states that realize the two spin states of the Hubbard model. A magnetic field $B$ is used to energetically split the hyperfine levels via the Zeeman effect. Tunneling between lattice sites $\tau$ is controlled by laser properties and short-ranged on-site fermion-fermion interactions $U$ arise because of s-wave scattering between atoms in different spin states. Importantly, the thermal energy of ultracold atoms as well as the energy scales associated with $\tau$, $U$ and $B$ can all be chosen to be much smaller than the gap between Bloch bands hence limiting the dynamics to a single band Hubbard model.

When immersing the optical lattice into a Bose-Einstein condensate additional scattering processes between the lattice atoms and the background Bose-Einstein condensate occur \cite{Klein}. The dominant effect of these additional interactions is local pure dephasing $\gamma_j$ of the lattice atoms described by the operators $L_j$ in the main text (independently of whether the lattice atoms are bosons or fermions). The background gas may also cause small renormalizations of the parameters appearing in the Hubbard model but these are often negligible. Experimentally, such interactions have recently been studied for individual ${}^{133}$Cs atoms immersed in ${}^{87}$Rb condensates paving the way to controlled atom-quantum bath interactions \cite{Widera}.

Our setup requires that the dephasing interaction affects both hyperfine states of the fermionic lattice atoms equally. This spin-agnostic interaction can for instance be realized by forming the Bose-Einstein condensate out of spin-0 bosons. Such two-component mixtures have recently been experimentally realized with fermionic ${}^{87}$Sr and spin-0 ${}^{88}$Sr \cite{Mickelson}. Also mixtures of fermionic ${}^{40} $K atoms and bosonic ${}^{87}$Rb were described using spin-independent Bose-Fermi interactions in Supplementary Reference \citen{Wang}. Finally, to realize the $\eta-$paring example a uniform potential lattice potential is required as was recently realized experimentally in Supplementary Reference \citen{Gaunt}.


\begin{thebibliography}{10}
		
		\bibitem{AlbertJiang1} Albert, V. V., \& Jiang, L. Symmetries and conserved quantities in Lindblad master equations. \emph{Phys. Rev. A} {\bf 89,} 022118 (2014).
		
		\bibitem{Bellomo} Bellomo, B., Giorgi, G.L., Palma, G.M., Zambrini, R. Quantum synchronization as a local signature of super- and subradiance. \emph{Phys. Rev. A} {\bf 95,} 043807 (2017). 
		
		\bibitem{Lee} Lee, T.E. \& Sadeghpour, H. R. Quantum Synchronization of Quantum van der Pol Oscillators with Trapped Ions. \emph{Phys. Rev. Lett.} {\bf 111,} 234101 (2013).
		
		\bibitem{Xu} Xu, M, et al. Synchronization of two ensembles of atoms. \emph{Phys. Rev. Lett.}  {\bf 113,} 154101 (2014).
		
		\bibitem{AlbertJiang2} Albert, V. V., Bradlyn, B., Fraas, M., \& Jiang L. Geometry and Response of Lindbladians. \emph{Phys. Rev. X} {\bf 6,} 041031 (2016).
		
		\bibitem{BN} Baumgartner, B. \& Narnhofer, H. Analysis of quantum semigroups with GKS-Lindblad generators: II. General. \emph{J. Phys. A: Math. Theor.} {\bf 41,} 395303 (2008). 	
		
		
		\bibitem{Karasik} Karasik, R.I, et al. Criteria for dynamically stable decoherence-free subspaces and incoherently generated coherences. \emph{Phys. Rev. A} {\bf 77,} 052301 (2008).	
		
		\bibitem{Brooke} Brooke, P.G., Cresser, J.D.,  \& Patra, M.K. Decoherence-free quantum information in the presence of dynamical evolution. \emph{Phys. Rev. A} {\bf 77,}, 062313 (2008).
		\bibitem{Wu} Wu, S.L., Wang, L.C., \& Yi, X.X. Time-dependent Decoherence-Free Subspace. \emph{J. Phys. A: Math. Theor. }{\bf 45,} 405305 (2012).
		
		\bibitem{Daley} Daley A. Quantum Trajectories and open many-body quantum systems, \emph{Adv. Phys.} {\bf 63,} 77 (2014).
		
		\bibitem{Kraus} Kraus, B., B\"{u}chler, H. P., Diehl, S., Kantian, A., Micheli,  A., \&  Zoller, P. Preparation of entangled states by quantum Markov processes. \emph{Phys. Rev. A} {\bf 78,} 042307 (2008).
		
		\bibitem{Jordi} Fern\'{a}ndez-Hurtado, V., Mur-Petit, J., Garc\'{i}a-Ripoll, J. J., \& Molina, R. A. Lattice scars: surviving in an open discrete billiard. \emph{New J. Phys.} {\bf 16,} 035005 (2014)
		
		\bibitem{Turner} Turner, C.J., et al. Weak ergodicity breaking from quantum many-body scars. \emph{Nat. Phys.} {\bf 14,} 745 (2018). 
		
		\bibitem{Yang} Yang, C.N. $\eta$-pairing and off-diagonal long-range order in a Hubbard model. \emph{Phys. Rev. Lett.} {\bf 63,} 2144 (1989). 
		

		
	
		
		\bibitem{Esslinger} Esslinger, T. Fermi-Hubbard Physics with Atoms in an Optical Lattice. \emph{Annu. Rev. Condens. Matter Phys.} {\bf 1,} 129 (2010). 
		
		\bibitem{Klein} Klein A., Bruderer, M., Clark, S.R., \& Jaksch D. Dynamics, dephasing and clustering of impurity atoms in Bose-Einstein condensates. \emph{New J. Phys.} {\bf 9,} 411 (2007).
		
		\bibitem{Widera} D. Mayer, et al. Controlled doping of a bosonic quantum gas with single neutral atoms. Preprint at http://arxiv.org/abs/1805.01313 (2018).
	
		

		\bibitem{Mickelson} Mickelson, P.G., et al. Bose-Einstein condensation of ${}^{88} Sr$ through sympathetic cooling with ${}^{87} Sr$. \emph{Phys. Rev. A} {\bf 81,} 051601(R) (2010).
		
		\bibitem{Wang}  Wang, D.-W., Lukin, M.D., Demler, E.  Engineering Superfluidity in Bose-Fermi Mixtures of Ultracold Atoms. \emph{Phys. Rev. A} {\bf 72,} 051604(R) (2005)
		
		\bibitem{Gaunt} Gaunt, A.L., et al. Bose-Einstein condensation of atoms in a uniform potential. \emph{Phys. Rev. Lett.} {\bf 110,} 200406 (2013).
		

		
	\end{thebibliography}
\end{document}